\documentclass[english]{cccconf}

\usepackage{lineno,hyperref}
\modulolinenumbers[5]

\usepackage[english]{babel}
\usepackage{theorem}
\usepackage[framemethod=tikz]{mdframed}
\theoremheaderfont{\itshape\bfseries}
{\theorembodyfont{\itshape}
	\newtheorem{assumption}{\textbf{Assumption}}
	
	\newtheorem{definition}{\textbf{Definition}}
	\newtheorem{theorem}{\textbf{Theorem}}
	
	\newtheorem{remark}{\textbf{Remark}}
	
	\newtheorem{problem}{\textbf{Problem}}
}

\usepackage{newtxtext,newtxmath}
\usepackage{amsmath}
\usepackage{amsfonts}
\usepackage{mathrsfs}
\usepackage{xcolor}
\usepackage{graphicx}
\usepackage[ruled,vlined]{algorithm2e}
\usepackage[most]{tcolorbox}
\usepackage{color,soul}
\usepackage{multirow}
\usepackage{graphicx}
\usepackage{subcaption}
\usepackage{mwe}

\newcommand{\T}{^{\mbox{\tiny T}}}
\newcommand{\R}{\mathbb{R}}

\newcommand{\eps}{\varepsilon}
\let\leq\leqslant
\let\geq\geqslant

\newenvironment{proof}[1][Proof]%
{\par\noindent\textit{#1:\ }}%
{\hspace*{\fill} \rule{6pt}{6pt}}
\newenvironment{proof*}[1][Proof]%
{\par\noindent\textit{#1:\ }}{}

\DeclareMathOperator{\diag}{diag}

\DeclareMathOperator{\rank}{rank}

\usepackage{tikz}
\usetikzlibrary{shapes,calc,arrows,patterns,decorations.pathmorphing
	,decorations.markings}
\usetikzlibrary{arrows.meta}

\newenvironment{system}[1]%
{\setlength{\arraycolsep}{0.5mm}\left\{ \; \begin{array}{#1}}%
	{\end{array} \right.}
\newenvironment{system*}[1]%
{\setlength{\arraycolsep}{0.5mm} \begin{array}{#1}}%
	{\end{array}}

\usepackage{tikz}
\usetikzlibrary{shapes,calc,arrows,patterns,decorations.pathmorphing
	,decorations.markings}
\usetikzlibrary{arrows.meta}

\begin{document}

	\title{$H_\infty$ Almost Output and Regulated Output Synchronization of Heterogeneous Multi-agent Systems: A Scale-free Protocol Design}

	\author{Donya Nojavanzadeh\aref{WSU},
	Zhenwei Liu\aref{ZW},
	Ali Saberi\aref{WSU},
	Anton A. Stoorvogel\aref{UT}}

\affiliation[WSU]{School of Electrical Engineering and Computer Science, Washington State University, Pullman,~WA,~USA
	\email{donya.nojavanzadeh@wsu.edu}\email{saberi@wsu.edu}}
\affiliation[ZW]{College of Information Science and Engineering, Northeastern University, Shenyang, P.~R.~China
	\email{liuzhenwei@ise.neu.edu.cn}}
\affiliation[UT]{Department of Electrical Engineering, Mathematics
	and Computer Science, University of Twente,~Enschede,~The Netherlands
	\email{a.a.stoorvogel@utwente.nl}}

	\maketitle

\begin{abstract}
	This paper studies scale-free protocol design for $H_{\infty}$ almost output and regulated output synchronization of heterogeneous multi-agent systems with linear, right-invertible and introspective agents
	in presence of external disturbances. The collaborative linear dynamic protocol designs are based on localized
	information exchange over the same communication network, which do not require any knowledge of the
	directed network topology and spectrum of the associated
	Laplacian matrix. Moreover, the proposed scale-free protocols achieve $H_{\infty}$ almost synchronization with a given
	arbitrary degree of accuracy for any size of the network.
\end{abstract}

	
	\footnotetext{This work is supported by the Nature Science Foundation of Liaoning Province, PR China under Grant 2019-MS-116, the Fundamental Research Funds for the Central Universities of China under Grant N2004014 and the United States National Science Foundation under Grant 1635184.}

\section{Introduction}

In recent decades, the synchronization problem for multi-agent systems
(MAS) has attracted substantial attention due to the wide potential
for applications in several areas such as automotive vehicle control,
satellites/robots formation, sensor networks, and so on. See books \cite{ren-book,wu-book, kocarev-book, bullobook} and references therein.

State synchronization inherently requires homogeneous networks
(i.e. agents which have identical models). 
 For heterogeneous network it is more reasonable to consider output synchronization since the dimensions of states and their physical interpretation may be different. If the agents have absolute measurements of their own dynamics in addition to relative information from the network, they are said to be introspective, otherwise, they are called non-introspective. For heterogeneous MAS with non-introspective agents, it is well-known that one needs to regulate outputs of the agents to a priori given trajectory generated by a so-called exosystem (see
\cite{wieland-sepulchre-allgower, grip-saberi-stoorvogel3}). Other works on synchronization of MAS with non-introspective agents can be found in the literature as  \cite{grip-yang-saberi-stoorvogel-automatica,grip-saberi-stoorvogel}. 

Synchronization and almost synchronization in presence of external disturbances are studied in the literature, where three classes of disturbances have been considered namely:
\begin{enumerate}
	 \item Disturbances and measurement noise with known frequencies.
	 \item Deterministic disturbances with finite power.
	 \item Stochastic disturbances with bounded variance.
	 \end{enumerate}
  
For disturbances and measurement noise with known frequencies, it is shown in \cite{zhang-saberi-stoorvogel-ACC2015} and \cite{zhang-saberi-stoorvogel-CDC2016} that actually exact synchronization is achievable. This is shown in \cite{zhang-saberi-stoorvogel-ACC2015} for heterogeneous MAS with minimum-phase and non-introspective agents and networks with time-varying directed communication graphs. Then, \cite{zhang-saberi-stoorvogel-CDC2016} extended this results for non-minimum phase agents utilizing localized information exchange. 

For deterministic disturbances with finite power, the notion of $H_\infty$ almost synchronization is introduced by Peymani et.al for homogeneous MAS with non-introspective agents utilizing additional communication exchange \cite{peymani-grip-saberi}. The goal of $H_\infty$ almost synchronization is to reduce the impact of disturbances on the synchronization error to an arbitrarily degree of accuracy (expressed in the $H_\infty$ norm). This work was extended later in \cite{peymani-grip-saberi-wang-fossen,zhang-saberi-grip-stoorvogel,zhang-saberi-stoorvogel-sannuti2} to heterogeneous MAS with non-introspective agents and without the additional communication and for network with time-varying graphs. $H_\infty$ almost synchronization via static protocols is studied in
\cite{stoorvogel-saberi-liu-nojavanzadeh-ijrnc19} for MAS with passive and passifiable agents. Recently, necessary and sufficient conditions are provided in \cite{stoorvogel-saberi-zhang-liu-ejc} for solvability $H_\infty$ almost synchronization of homogeneous networks with non-introspective agents and without additional communication exchange. Finally, we developed a scale-free framework for $H_\infty$ almost state synchronization for homogeneous network \cite{liu-saberi-stoorvogel-donya-almost-automatica} utilizing suitably designed localized information exchange.

In the case of stochastic disturbances with bounded variance, the concept of stochastic almost synchronization is introduced by \cite{zhang-saberi-stoorvogel-stochastic2} where both stochastic disturbance and disturbance with known frequency are considered. The idea of stochastic almost synchronization is to reduce the stochastic RMS norm of synchronization error arbitrary small in the presence of colored stochastic disturbances that can be modeled as the output of linear time invariant systems driven by white noise with unit power spectral intensities. By augmenting this model with agent model one can essentially assume that stochastic disturbance is white noise with unit power spectral intensities. In this case under linear protocols the stochastic RMS norm of synchronization error is the $H_2$ norm of the transfer function from disturbance to the synchronization error. As such one can formulate the stochastic almost synchronization equivalently in a deterministic framework requiring to reduce the $H_2$ norm of the transfer function from disturbance to synchronization error arbitrary small. This deterministic approach is referred to as almost $H_2$ synchronization problem which is equivalent to stochastic almost synchronization problem. Recent work on $H_2$ almost synchronization problem is \cite{stoorvogel-saberi-zhang-liu-ejc} which provided necessary and sufficient conditions for solvability of $H_\infty$ almost synchronization for homogeneous networks with non-introspective agents and without additional communication exchange. Finally, $H_2$ almost synchronization via static protocols is studied in \cite{stoorvogel-saberi-liu-nojavanzadeh-ijrnc19} for MAS with passive and passifiable agents.

In this paper, we develop scale-free protocol design for $H_{\infty}$ almost synchronization of heterogeneous MAS in presence of external disturbances. A collaborative linear parameterized dynamic protocols with localized information exchange is proposed which can work for heterogeneous MAS with any size of the network. The main contribution of this work is that the protocol design does not require any information about the communication network such as a lower bound of non-zero eigenvalues of the associated Laplacian matrix and the number of agents. Moreover, the scalable protocol achieves $H_{\infty}$ almost synchronization with a given arbitrary degree of accuracy for heterogeneous MAS with any number of agents.

\subsection*{Notations and Background}
Given a matrix $A\in \mathbb{R}^{m\times n}$, $A\T$ and $A^*$ denote
transpose and conjugate transpose of $A$ respectively while $\|A\|_2$
denotes the induced 2-norm (which has submultiplicative property). The $\mathrm{im}(\cdot)$ denote the image of matrix (vector). A square matrix $A$ is said to be Hurwitz stable if all its eigenvalues are in the open left half complex plane. $A\otimes B$ depicts the Kronecker product between $A$ and $B$. A block diagonal matrix constructed by $A_i$'s is denoted by $\diag\{A_i\}$ for $i=1,\hdots, n$. 
$I_n$ denotes the $n$-dimensional identity matrix and $0_n$ denotes $n\times n$ zero
matrix; sometimes we drop the subscript if the dimension is clear from
the context. For a deterministic continuous-time vector signal $v(t)$, the $L_2$ norm is defined by 
\begin{equation}
\|v(t)\|_{L_2}=\left(\int_{0}^{T}v(t)\T v(t)dt\right)^\frac{1}{2}
\end{equation}
and its \emph{Root Mean Square (RMS)} value is defined by 
\begin{equation}
\|v(t)\|_{RMS}=\left(\lim_{T\to \infty}\frac{1}{T}\int_{0}^{T}v(t)\T v(t)dt\right)^\frac{1}{2}
\end{equation}   
and for a stochastic vector signal $v(t)$ which is modeled as wide-sense stationary stochastic process, the $\|v(t)\|_{RMS}$ is given by
\begin{equation}\label{rms-s}
\|v(t)\|_{RMS}=\left( \mathbf{E}[v\T(t)v(t)]\right)^\frac{1}{2}
\end{equation}
where $\mathbf{E}[\cdot]$ stands for the expectation operation. For stochastic signals that approach
wide-sense stationarity as time $t$ goes to infinity (i.e. for asymptotically wide-sense stationary signals) \eqref{rms-s} is rewritten as 
\begin{equation}\label{rms-s-s}
\|v(t)\|_{RMS}=\left( \lim_{t\to \infty}\mathbf{E}[v\T(t)v(t)]\right)^\frac{1}{2}.
\end{equation}

The $H_\infty$ norm of $G(s)$ is defined as 
\[
\|G\|_{H_\infty}:=\sup_\omega \sigma_{\max}[G(j\omega)]
\]
where $\sigma_{\max}$ is the largest singular value of $G(j\omega)$. Let $\omega(t)$ and $z(t)$ be energy signals which are respectively the input and the corresponding output of the given system. Then, the $H_\infty$ norm of $G(s)$ turns out to coincide with its RMS gain, namely 
\[
\|G\|_{H_\infty}=\|G\|_{RMS\text{ }gain} =\sup_{\|\omega\|\neq 0} \frac{\|z\|_{RMS}}{\|\omega\|_{RMS}}
\]

An important property of the $H_\infty$ norm is that it is submultiplicative. That is for transfer functions $G_1$ and $G_2$, we have
\[
\|G_1G_2\|_{H_\infty} \leq\|G_1\|_{H_\infty}\|G_2\|_{H_\infty}.
\]

 A \emph{weighted graph} $\mathcal{G}$ is defined by a triple
$(\mathcal{V}, \mathcal{E}, \mathcal{A})$ where
$\mathcal{V}=\{1,\ldots, N\}$ is a node set, $\mathcal{E}$ is a set of
pairs of nodes indicating connections among nodes, and
$\mathcal{A}=[a_{ij}]\in \mathbb{R}^{N\times N}$ is the weighting
matrix. Each pair in $\mathcal{E}$ is called an \emph{edge}, where
$a_{ij}>0$ denotes an edge $(j,i)\in \mathcal{E}$ from node $j$ to
node $i$ with weight $a_{ij}$. Moreover, $a_{ij}=0$ if there is no
edge from node $j$ to node $i$. We assume there are no self-loops,
i.e.\ we have $a_{ii}=0$. A \emph{path} from node $i_1$ to $i_k$ is a
sequence of nodes $\{i_1,\ldots, i_k\}$ such that
$(i_j, i_{j+1})\in \mathcal{E}$ for $j=1,\ldots, k-1$. A
\emph{directed tree} with root $r$ is a subgraph of the graph
$\mathcal{G}$ in which there exists a unique path from node $r$ to
each node in this subgraph. A \emph{directed spanning tree} is a
directed tree containing all the nodes of the graph. 

For a weighted graph $\mathcal{G}$, the matrix
$L=[\ell_{ij}]$ with
\[
\ell_{ij}=
\begin{system}{cl}
\sum_{k=1}^{N} a_{ik}, & i=j,\\
-a_{ij}, & i\neq j,
\end{system}
\]
is called the \emph{Laplacian matrix} associated with the graph
$\mathcal{G}$. The Laplacian matrix $L$ has all its eigenvalues in the
closed right half plane and at least one eigenvalue at zero associated
with right eigenvector $\textbf{1}$, i.e. a vector with all entries
equal to $1$. When graph contains a spanning tree,  then it follows from \cite[Lemma
$3.3$]{ren-beard} that the Laplacian matrix $L$ has a simple
eigenvalue at the origin, with the corresponding right eigenvector
$\textbf{1}$, and all the other eigenvalues are in the open right-half
complex plane.

\section{Problem formulation}

Consider a heterogeneous MAS composed of $N$ nonidentical linear time-invariant agents of the form,
\begin{equation}\label{homst-agent-model}
\begin{system*}{cl}
\dot{x}_i &= A_ix_i +B_i u_i+E_i\omega_i,  \\
y_i &= C_ix_i,
\end{system*}\qquad (i=1,\ldots,N)
\end{equation}
where $x_i\in\R^{n_i}$, $u_i\in\R^{m_i}$, $y_i\in\R^p$ are respectively the
state, input, and output vectors of agent $i$, and
$\omega_i\in\R^{w_i}$ is the external disturbance. $n_{q0}\geq 1$ is the  upper bound on infinite zeros of triples $(C_i,A_i,B_i)$ for $i \in \{1,..., N\}$. 

The agents are introspective, meaning that each agent has access to its own local information. Specifically, each agent has access to part of its state, i.e.,
\begin{equation}
z_i=C_i^mx_i
\end{equation}
where $z_i\in\mathbb{R}^{q_i}$.

The communication network provides each agent with a linear
combination of its own outputs relative to that of other neighboring
agents. In particular, each agent $i\in\{1,\ldots,N\}$ has access to the
quantity
\begin{equation}\label{homst-zeta1}
\zeta_i = \sum_{j=1}^{N}a_{ij}(y_i-y_j)
\end{equation}
where $a_{ij}\geq 0$ and $a_{ii}=0$ indicate the communication among
agents. This communication topology of the network can
be described by a weighted and directed graph $\mathcal{G}$ with nodes
corresponding to the agents in the network and the weight of edges
given by the coefficient $a_{ij}$. In terms of the coefficient of the associated
Laplacian matrix $L$, \eqref{homst-zeta1} can be rewritten as
\begin{equation}\label{homst-zeta}
\zeta_i = \sum_{j=1}^{N}\ell_{ij}y_j.
\end{equation}



The goal of this paper is to design scale-free protocols which can be achieved by utilizing localized information exchange among neighbors, as such each agent $i\in\{1,\ldots, N\}$ has access to localized information, denoted by $\hat{\zeta}_i$, of the form
\begin{equation}\label{eqa1}
\hat{\zeta}_i=\sum_{j=1}^Na_{ij}(\xi_i-\xi_j)
\end{equation}
where $\xi_j$ is a variable produced internally by
agent $j$ which will be appropriately chosen in the coming sections.

For almost output synchronization, we define the set of graphs $\mathbb{G}^N$ for the network communication topology as following.

\begin{definition}\label{def-G}
	Let $\mathbb{G}^N$ denote the set of directed graphs of $N$
	agents which contains a directed spanning tree.
\end{definition}


We formulate the scale-free $H_\infty$ almost output
synchronization problem of a heterogeneous MAS with localized information exchange as following.

\begin{problem}\label{prob4}
  The \textbf{scale-free $H_\infty$ almost output synchronization
    problem with localized information exchange (scale-free $H_\infty$-AOSWLIE)} for MAS \eqref{homst-agent-model} and
  \eqref{homst-zeta} is to find, if possible, a scalar parameter $\eps^*>0$ and a fixed linear protocol
  parameterized in scalar $\eps$,  \textbf{only} based on the knowledge of agent models $(C_i, A_i,B_i)$, of the form
  \begin{equation}\label{protoco1}
  \begin{system}{cl}
  \dot{x}_{c_i}&=A_{c_i}(\eps)x_{c_i}+B_{c_i}(\eps){\zeta}_i
  +C_{c_i}(\eps)\hat{\zeta}_i+D_{c_i}(\eps)z_i\\
  u_i&=E_{c_i}(\eps)x_{c_i}+F_{c_i}(\eps){\zeta}_i
  +G_{c_i}(\eps)\hat{\zeta}_i+H_{c_i}(\eps)z_i
  \end{system}
  \end{equation}
  where $\hat{\zeta}_i$ is defined by \eqref{eqa1}, with $\xi_i=M_{c_i} x_{c_i}$ with $x_{c_i}\in \mathbb{R}^{n_{c_i}}$ such that for any number of agents $N$, and any communication graph $\mathcal{G}\in\mathbb{G}^N$ we have:
  \begin{itemize}
  	\item in the absence of the disturbance $\omega=\begin{pmatrix} \omega_1\T &\hdots & \omega_N\T 
  	\end{pmatrix}\T$, for all initial
  	conditions the output synchronization 
  	 \begin{equation}\label{synch_org}
  	\lim_{t\to \infty} (y_i-y_j)=0 \quad  \text{for all $i,j \in \{1,...,N\}$}
  	\end{equation}
  	 is achieved
  	for any $\eps\in(0,\eps^*]$.
  	\item in the presence of the disturbance $\omega$, for any $\gamma>0$,
  	one can render the $H_{\infty}$ norm from $\omega$ to $y_i-y_j$ less
  	than $\gamma$ by choosing $\eps$ sufficiently small.
  \end{itemize}

The architecture of the protocol \eqref{protoco1} is shown in Figure \ref{Architecture}.
\end{problem}

\begin{figure}[t]
	\includegraphics[width=8.5cm, height=5cm]{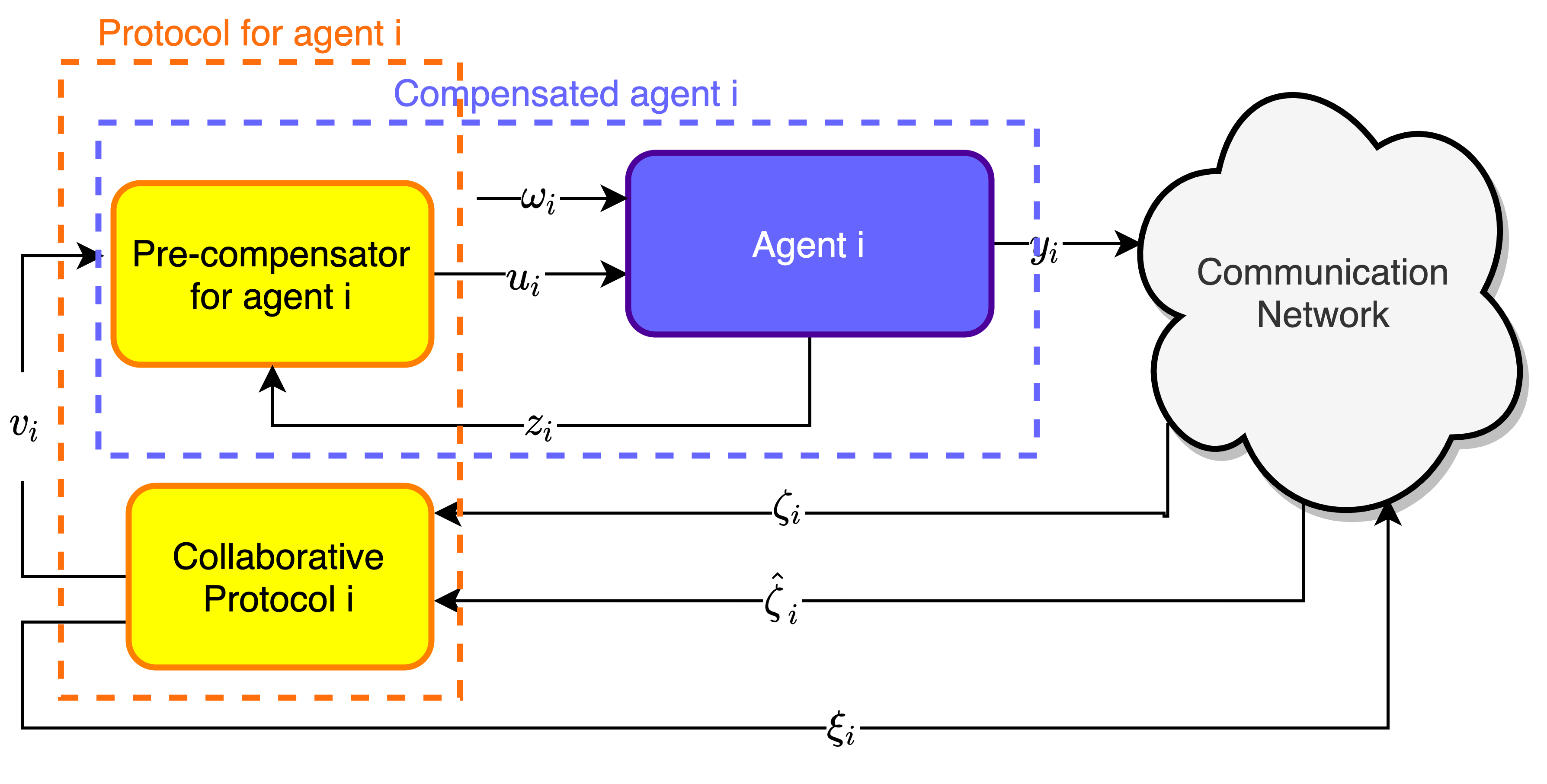}
	\centering
	\caption{Architecture of scale-free collaborative protocol for $H_{\infty}$ almost output synchronization}\label{Architecture}
\end{figure}

In the case of $H_\infty$ almost regulated output synchronization, we consider a reference trajectory $y_r$ generated by a so-called exosystem as
\begin{equation}\label{exo}
\begin{system*}{cl}
\dot{x}_r&=A_rx_r, \quad x_r(0)=x_{r0},\\
y_r&=C_rx_r,
\end{system*}
\end{equation}
where $x_r \in\mathbb{R}^r$ and $y_r\in\mathbb{R}^p$.

We assume a nonempty subset $\mathscr{C}$ of the agents which have access to their output relative to the output of the exosystem. In other word, each agent $i$ has access to the quantity 
\begin{equation}
\Psi_i=\iota_i(y_i-y_r), \qquad \iota_i=\begin{system}{cl}
1, \quad i\in \mathscr{C},\\
0, \quad i\notin \mathscr{C}.
\end{system}
\end{equation}

By combining the above equation with \eqref{homst-zeta1}, the information exchange among agents is given by
\begin{equation}\label{zetabar}
\bar{\zeta}_i=\sum_{j=1}^{N}a_{ij}(y_i-y_j)+\iota_i(y_i-y_r).
\end{equation}

The exchanging information $\bar{\zeta}_i$ as defined in above, can be rewritten in terms of the coefficients of so-called expanded Laplacian matrix $\bar{L}=L+diag\{\iota_i\}=[\bar{\ell}_{ij}]_{N \times N}$ as
\begin{equation}\label{zetabar2}
\bar{\zeta}_i=\sum_{j=1}^{N}\bar{\ell}_{ij}(y_j-y_r).
\end{equation}

Note that $\bar{L}$ is not a regular Laplacian matrix associated to the graph, since the sum of its rows need not be zero. We know that all the eigenvalues of $\bar{L}$, have positive real parts. In particular matrix $\bar{L}$ is invertible.

To guarantee that each agent gets the information from the exosystem, we need to make sure that there exists a path from node set $\mathscr{C}$ to each node.  Therefore, we define the following set of graphs.
\begin{definition}\label{def_rootset}
	Given a node set $\mathscr{C}$, we denote by $\mathbb{G}_{\mathscr{C}}^N$ the set of all graphs with $N$ nodes containing the node set $\mathscr{C}$, such that every node of the network graph $\mathcal{G}\in\mathbb{G}_\mathscr{C}^N$ is a member of a directed tree
	which has its root contained in the node set $\mathscr{C}$. We will refer to the node set $\mathscr{C}$ as root set.
\end{definition}

\begin{remark}
	Note that Definition \ref{def_rootset} does not require necessarily the existence of directed spanning tree.
\end{remark}

Now, we formulate scale-free $H_{\infty}$ almost regulated output synchronization problem.

\begin{problem}\label{prob-reg}
	The \textbf{scale-free $H_\infty$ almost regulated output synchronization
		problem with localized information exchange (scale-free $H_\infty$-AROSWLIE)} for MAS \eqref{homst-agent-model} and
	\eqref{zetabar2} and the associated exosystem \eqref{exo} is to find, if possible, a scalar parameter $\eps^*>0$ and a fixed linear protocol
	parameterized in scalar $\eps$,  \textbf{only} based on the knowledge of agent models $(C_i, A_i,B_i)$, of the form
	\begin{equation}\label{protoco2}
	\begin{system}{cl}
	\dot{x}_{c_i}&=A_{c_i}(\eps)x_{c_i}+B_{c_i}(\eps)\bar{\zeta}_i
	+C_{c_i}(\eps)\hat{\zeta}_i+D_{c_i}(\eps)z_i\\
	u_i&=E_{c_i}(\eps)x_{c_i}+F_{c_i}(\eps)\bar{\zeta}_i
	+G_{c_i}(\eps)\hat{\zeta}_i+H_{c_i}(\eps)z_i
	\end{system}
	\end{equation}
	where $\hat{\zeta}_i$ is defined by \eqref{eqa1}, with $\xi_i=M_{c_i} x_{c_i}$ with $x_{c_i}\in \mathbb{R}^{n_{c_i}}$ such that for any number of agents $N$, and any communication graph $\mathcal{G}\in\mathbb{G}_\mathscr{C}^N$ we have:
	\begin{itemize}
		\item in the absence of the disturbance $\omega=\begin{pmatrix} \omega_1\T &\hdots & \omega_N\T 
		\end{pmatrix}\T$, for all initial
		conditions the regulated output synchronization 
		\begin{equation}\label{synch_org-reg}
		\lim_{t\to \infty} (y_i-y_r)=0 \quad  \text{for all $i \in \{1,...,N\}$}
		\end{equation}
		is achieved
		for any $\eps\in(0,\eps^*]$.
		\item in the presence of the disturbance $\omega$, for any $\gamma>0$,
		one can render the $H_{\infty}$ norm from $\omega$ to $y_i-y_r$ less
		than $\gamma$ by choosing $\eps$ sufficiently small.
	\end{itemize}

	The architecture of the protocol \eqref{protoco2} is shown in Figure \ref{Architecture-reg}.
\end{problem}

\begin{figure}[t]
	\includegraphics[width=8.5cm, height=5.5cm]{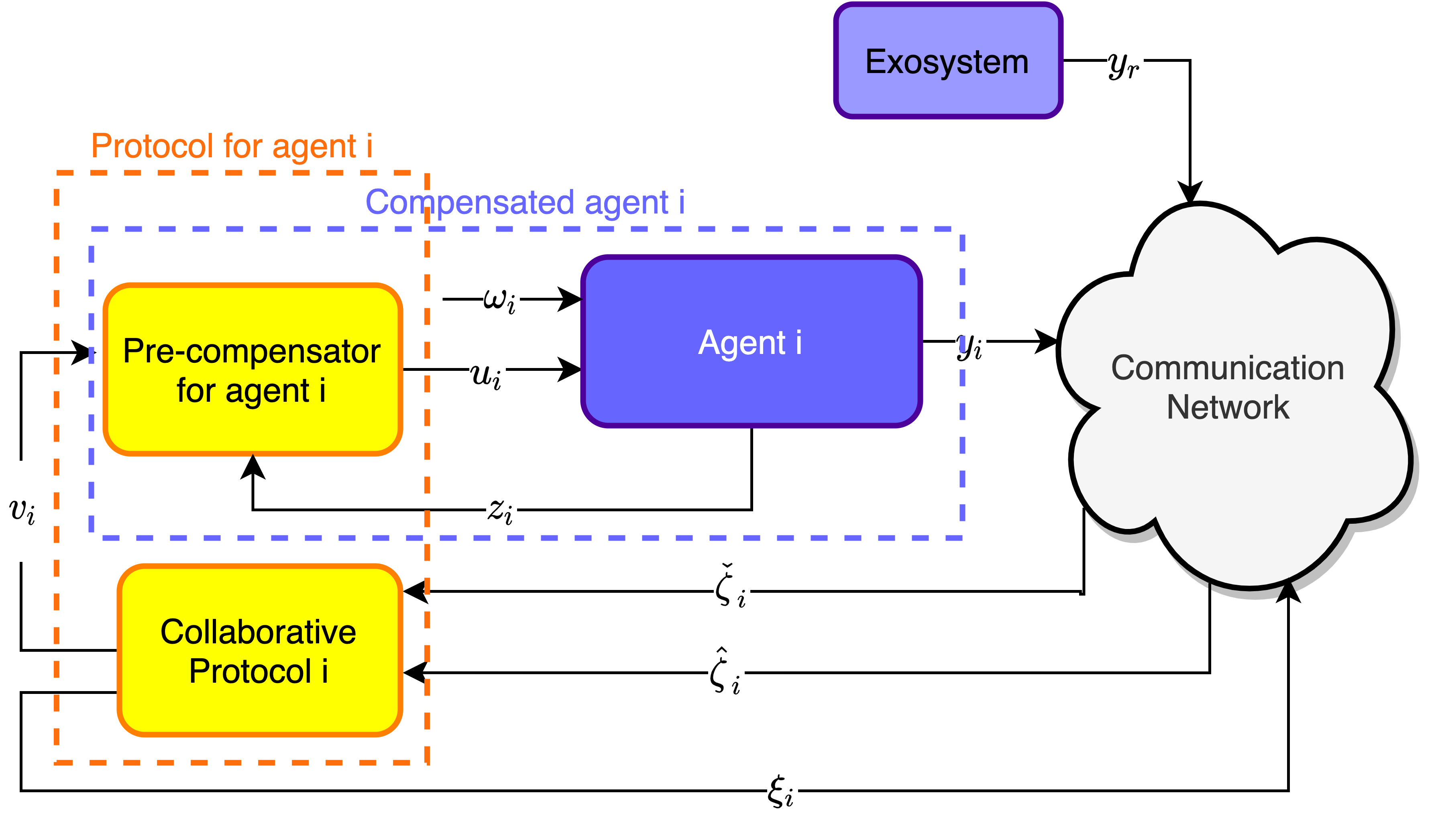}
	\centering
	\caption{Architecture of scale-free collaborative protocol for $H_{\infty}$ almost regulated output synchronization}\label{Architecture-reg}
\end{figure}

We make the following assumptions for
agents and the exosystem.
\begin{assumption}\label{ass2} For agents $i \in \{1,..., N\}$,
	\begin{enumerate}
		\item $(C_i,A_i,B_i)$ is stabilizable and detectable.
		\item $(C_i,A_i,B_i)$ is right-invertible.
		\item $(C_i^m, A_i)$ is detectable.
	\end{enumerate}
\end{assumption}

\begin{assumption}\label{ass1} For exosystem
	\begin{enumerate}
		\item $(C_r,A_r)$ is observable.
		\item All the eigenvalues of $A_r$ are in the closed left half complex plane.
	\end{enumerate}
\end{assumption}

\section{Scale-free $H_\infty$ almost output synchronization}

According to Figure \ref{Architecture}, the design for scale-free $H_{\infty}$ almost output synchronization includes two major modules. First, we design precompensators for homogenizing agents, then we design collaborative protocols for the compensated agents to achieve $H_{\infty}$ almost output synchronization. The design consists of the following steps.

\hspace{3mm} \textbf{Step I: choosing target model} 
First we choose suitable target model, i.e. $(C_d, A_d,B_d)$ such that the following conditions are satisfied.
\begin{enumerate}
	\item $\rank{C_d}=p$,
	\item $(C_d, A_d,B_d)$ is invertible, of uniform rank $n_q\geq n_{q0}$ and has no invariant zero. 	$n_{q0}\geq 1$ is the maximum order of infinite zeros of triples $(C_i,A_i,B_i)$ for $i \in \{1,..., N\}$. 
	\item eigenvalues of $A_d$ are in the closed left half plane.
\end{enumerate}
\begin{remark}
	The choice of $A_d$ is designer choice and obviously the eigenvalues of $A_d$ determines the desired trajectory.
\end{remark}
\begin{remark}\label{remark-ABC}
	Without loss of generality, we assume that the triple $(C_d,A_d,B_d)$ has the following form:
	\[
	\begin{system*}{cll}
	A_d&=\bar{A}+\bar{B}\Gamma,\quad \bar{A}:=\begin{pmatrix}
	0&I_{p(n_q-1)}\\0&0	\end{pmatrix},\\ B_d&=\bar{B}:=\begin{pmatrix}
	0\\I_p
	\end{pmatrix},
	 C_d=\bar{C}:=\begin{pmatrix}
	I_p&0
	\end{pmatrix}.
	\end{system*}
	\]
\end{remark}

\hspace{3mm} \textbf{Step II: designing precompensators}  Following the instruction given in \cite[Appendix B]{yang-saberi-stoorvogel-grip-journal}, we design the following precompensator for each agent of MAS \eqref{homst-agent-model} to homogenize the agents to the target model chosen in the previous step.
\begin{tcolorbox}[colback=white]
	\begin{equation}\label{precom}
		\begin{system*}{cl}
			\dot{p}_i&=G_ip_i+H_{1i}v_i+H_{2i}z_{i}, \\
			u_i&=Q_ip_i+R_{1i}v_i+R_{2i}z_{i},
		\end{system*}
	\end{equation}
	where $v_i$ is the input of the precompensator.
\end{tcolorbox}

We obtain the compensated agents by combining \eqref{homst-agent-model} and \eqref{precom} as
\begin{equation}\label{comMAS}
	\begin{system*}{cl}
		\dot{\bar{x}}_i&=\bar{A}\bar{x}_i+\bar{B}(v_i+\Gamma\bar{x}_i)+\bar{E}_i\bar{\omega}_i, \\
		y_i&=\bar{C}\bar{x}_i,
	\end{system*}
\end{equation}
where $
\bar{\omega}_i=\begin{pmatrix}
	\omega_i\T&\rho_i\T
\end{pmatrix}\T$. 
Meanwhile, $\rho_i$ is generated by the following system
\[
\begin{system*}{cl}
	\dot{\theta}_i&=S_i\theta_i+E_{0i}\bar{\omega}_i, \\
	\rho_i&=W_i\theta_i,
\end{system*}
\] 
where $S_i$ is Hurwitz stable.

\hspace{3mm} \textbf{Step III: designing collaborative protocols for the compensated agents} 
In this step, the following linear dynamic protocol is designed for the compensated agents \eqref{comMAS} as
\begin{tcolorbox}[colback=white]
	\begin{equation}\label{pro-lin-partial}
		\begin{system}{cll}
			\dot{\hat{x}}_i&=& \bar{A}\hat{x}_i-\eps^{-n_q}\bar{B}F\Delta\hat{\zeta}_{i}+\bar{B}\Gamma\hat{x}_i+\eps^{-1}K(\zeta_{i}-\bar{C}\hat{x}_i)\\
			\dot{\chi}_i &=& \bar{A}\chi_i+\bar{B}v_i+\bar{B}\Gamma\chi_i+\eps^{-1}(\hat{x}_i-\hat{\zeta}_{i})\\
			v_i &=& -\eps^{-n_q}F\Delta\chi_i,
		\end{system}
	\end{equation}
	where $\eps\in(0,\eps^*]$ is the tuning parameter. We define the scaling matrix $\Delta=\diag\{I_p,\eps I_p,\hdots, \eps^{n_q-1}I_p\}\in\mathbb{R}^{pn_q \times pn_q}$ and we choose matrix $F$ such that $\bar{A}-\bar{B}F$ is Hurwitz stable. Meanwhile, we partition $\bar{A}, \bar{B}, \Gamma$ as
	\[
	\bar{A}=\begin{pmatrix}
		0_{p\times p}&\bar{C}_1\\0_{p(n_q-1)\times p}&\bar{A}_1
	\end{pmatrix}, \bar{B}=\begin{pmatrix}
		0_{p\times p}\\\bar{B}_1
	\end{pmatrix}, \Gamma=\begin{pmatrix}\Gamma_1&\Gamma_2\end{pmatrix}
	\]
	with $\bar{C}_1=\begin{pmatrix}
		I_p&0&\cdots&0
	\end{pmatrix}$.
	We  choose $K=\begin{pmatrix}
		K_1\\K_2K_1
	\end{pmatrix}$
	with $K_2$ such that $\bar{A}_1+\bar{B}_1\Gamma_2-K_2\bar{C}_1$ is Hurwitz stable. And then let $K_1=K_1\T>\frac{\alpha}{2}I\in \R^{p\times p}$, where the value of $\alpha$ depends on $K_2$ and is given explicitly in the proof of Theorem \ref{sstem1}.
	
	In this protocol, each agent communicate $\chi_i$ with its neighbors, (i.e., $\xi_i=\chi_i$). Therefore, each agent has access to the localized information exchange appeared in \eqref{pro-lin-partial} represented by
	\begin{equation}\label{add_1}
		\hat{\zeta}_{i}=\sum_{j=1}^N{\ell}_{ij}\chi_j,
	\end{equation}
	which can be simply obtained by $\xi_i=\chi_i$ in \eqref{eqa1}.
\end{tcolorbox}
 We have the following theorem.

\begin{theorem}\label{sstem1}
	Consider a MAS described by \eqref{homst-agent-model} and  \eqref{homst-zeta} satisfying Assumption \ref{ass2}. Let $\mathbb{G}^N$ be the set of network graphs as
	defined in Definition \ref{def-G}.Then, the scalable $H_\infty$-AOSWLIE problem as stated in Problem \ref{prob4} is solvable. More specifically,
	\begin{enumerate}
		\item in the absence of the disturbance $\omega$, protocol \eqref{pro-lin-partial} and \eqref{precom} achieves output synchronization \eqref{synch_org}, for any graph $\mathcal{G}\in \mathbb{G}^N$ with any number of
		agents $N$, and for any
		$\eps\in(0,\eps^*]$ with $\eps^*<1$ where the value of $\eps^*$ only depends on $\bar{A},\bar{B},\Gamma, F$,
		\item in the presence of the disturbance
		$\omega$, for any $\gamma>0$, the $H_\infty$ norm from $\omega$ to
		$y_i-y_j$ is less that $\gamma$ for all $i,j \in\{1,\hdots, N\}$ by choosing $\eps$ sufficiently small.
	\end{enumerate}
\end{theorem}

\begin{proof}
	The proof includes four steps.
	
	1) \textit{Homogenization:}
	
	According to the construction given in \cite[Appendix B]{yang-saberi-stoorvogel-grip-journal}, we design precompensator \eqref{precom} for the heterogeneous MAS \eqref{homst-agent-model} to obtain the corresponding homogenized system \eqref{comMAS}.\\
	
	2) \textit{Closed loop system:}
	
	First, by defining $\tilde{x}_i=\bar{x}_i-\bar{x}_N$, $\tilde{y}_i=y_i-y_N$, $\tilde{\hat{x}}_i=\hat{x}_i-\hat{x}_N$, and $\tilde{\chi}_i=\chi_i-\chi_N$, we have
	\begin{align*}
		\dot{\tilde{x}}_i&=\bar{A}\tilde{x}_i+\bar{B}(v_i-v_N)+\bar{B}\Gamma\tilde{x}_i+\bar{E}_i\bar{\omega}_i-\bar{E}_N\bar{\omega}_N  \\
		\tilde{y}_i&=\bar{C}\tilde{x}_i\\
		\dot{\tilde{\hat{x}}}_i&= \bar{A}\tilde{\hat{x}}_i-\eps^{-n_q}\bar{B}F\Delta(\hat{\zeta}_{i}-\hat{\zeta}_{N})+\bar{B}\Gamma\tilde{\hat{x}}_i+\eps^{-1}K(\zeta_{i}-\zeta_{N}-\bar{C}\tilde{\hat{x}}_i)\\
		\dot{\tilde{\chi}}_i &= \bar{A}\tilde{\chi}_i+\bar{B}(v_i-v_N)+\bar{B}\Gamma\tilde{\chi}_i+\eps^{-1}(\hat{x}_i-\hat{x}_N-\hat{\zeta}_{i}+\hat{\zeta}_{N})\\
		v_i&-v_N = -\eps^{-n_q}F\Delta\tilde{\chi}_i,
	\end{align*}
	
	Let
	\[
	\tilde{x}=\begin{pmatrix}
		\tilde{x}_1\\\vdots\\\tilde{x}_{N-1}
	\end{pmatrix}, \tilde{\hat{x}}=\begin{pmatrix}
		\tilde{\hat{x}}_1\\\vdots\\\tilde{\hat{x}}_{N-1}
	\end{pmatrix}, \tilde{\chi}=\begin{pmatrix}
		\tilde{\chi}_1\\\vdots\\\tilde{\chi}_{N-1}
	\end{pmatrix}, \bar{\omega}=\begin{pmatrix}
		\bar{\omega}_1\\\vdots\\\bar{\omega}_{N}
	\end{pmatrix},
	\]
	then, we can rewrite the closed loop system as
	\begin{align*}
		\dot{\tilde{x}}&=(I\otimes(\bar{A}+\bar{B}\Gamma))\tilde{x}-\eps^{-n_q}(I\otimes(\bar{B}F\Delta))\tilde{\chi}+(\Pi\otimes I)\bar{E}\bar{\omega}\\
		\dot{\tilde{\hat{x}}}&=(I\otimes(\bar{A}+\bar{B}\Gamma-\eps^{-1}K\bar{C}))\tilde{\hat{x}}\\
		&\hspace{1cm}-\eps^{-n_q}(\hat{L}\otimes(\bar{B}F\Delta))\tilde{\chi}+\eps^{-1}(\hat{L}\otimes (K\bar{C}))\tilde{x}\\
		\dot{\tilde{\chi}}&= (I\otimes (\bar{A}+\bar{B}\Gamma)-\eps^{-1}\hat{L}\otimes I-\eps^{-n_q}I\otimes(\bar{B}F\Delta))\tilde{\chi}+\eps^{-1}\tilde{\hat{x}}
	\end{align*}
	where $\hat{L}=[\hat{\ell}_{ij}]_{(N-1)\times (N-1)}$ with $\hat{\ell}_{ij}=\ell_{ij}-\ell_{Nj}$ for $i,j=1,\cdots, N-1$ and $\Pi=\begin{pmatrix}
		I&-\textbf{1}
	\end{pmatrix}$. According to \cite[Lemma 1]{liu-saberi-stoorvogel-donya-almost-automatica}, we have all eigenvalues of $\hat{L}$ are the all nonzero eigenvalues of ${L}$. 
	
	Further, by defining $e=\tilde{x}-\tilde{\chi}$ and $\bar{e}=(\hat{L}\otimes I)\tilde{x}-\tilde{\hat{x}} $, we have
	\begin{align}
		\nonumber\dot{\tilde{x}}&=(I\otimes(\bar{A}-\eps^{-n_q}\bar{B}F\Delta+\bar{B}\Gamma))\tilde{x}\\
		&\hspace{1.5cm}+\eps^{-n_q}(I\otimes(\bar{B}F\Delta))e+(\Pi\otimes I)\bar{E}\bar{\omega}\label{systmat}\\
		\dot{e}&= (I\otimes \bar{A}-\eps^{-1}\hat{L}\otimes I)e+I\otimes(\bar{B}\Gamma)e+\eps^{-1}\bar{e}+(\Pi\otimes I)\bar{E}\bar{\omega}\label{systmat2}\\
		\dot{\bar{e}}&=(I\otimes(\bar{A}+\bar{B}\Gamma-\eps^{-1}K\bar{C}))\bar{e}+(\hat{L}\Pi\otimes I)\bar{E}\bar{\omega}
	\end{align}
	
	Let $\tilde{x}_\Delta=(I\otimes\Delta) \tilde{x}$, $e_\Delta=(I\otimes\Delta) e$, and $\bar{e}_\Delta=(I\otimes\bar{\Delta}) \bar{e}$ with 
	\[
	\bar{\Delta}=\begin{pmatrix}
		I_p&0\\-\eps K_2&\eps I_{p(n_q-1)}
	\end{pmatrix},
	\]
	then we have
	\begin{align}
		\nonumber\dot{\tilde{x}}_\Delta&=\eps^{-1}I\otimes(\bar{A}-\bar{B}F)\tilde{x}_\Delta+(I\otimes\bar{\Gamma})\tilde{x}_\Delta\\
		&\hspace{1cm}+\eps^{-1}(I\otimes(\bar{B}F))e_\Delta+(\Pi\otimes \Delta)\bar{E}\bar{\omega}\label{wd1}\\
		\nonumber\dot{e}_\Delta&= \eps^{-1}(I\otimes \bar{A}-\hat{L}\otimes I)e_\Delta+(I\otimes\bar{\Gamma})e_\Delta\\
		&\hspace{1cm}+\eps^{-1}I\otimes(\Delta\bar{\Delta}^{-1})\bar{e}_\Delta+(\Pi\otimes \Delta)\bar{E}\bar{\omega}\label{wd3}\\
		\dot{\bar{e}}_\Delta&=(I\otimes \bar{A}_{\bar{e}})\bar{e}_\Delta+(\hat{L}\Pi\otimes \bar{\Delta})\bar{E}\bar{\omega}\label{wd2}
	\end{align}
	where $\bar{\Gamma}=\eps^{n_q-1}\bar{B}\Gamma\Delta^{-1}$ and
	\[
	\bar{A}_{\bar{e}}=\begin{pmatrix}
		-\eps^{-1}K_1+\bar{C}_1K_2&\eps^{-1}\bar{C}_1\\
		\eps \Psi&\bar{A}_1+\bar{B}_1\Gamma_2-K_2\bar{C}_1
	\end{pmatrix}
	\]
	with $\Psi=\bar{A}_1K_2-K_2\bar{C}_1K_2+\bar{B}_1\Gamma_1+\bar{B}_1\Gamma_2K_2$.\\

	3) \textit{Synchronization in the absence of disturbance:}	
	
	In this step, we consider the stability of system \eqref{wd1}-- \eqref{wd2} by setting the disturbance equal to zero, i.e., $\bar{\omega}=0$, which implies synchronization in the absence of disturbance
	\begin{align}
		\dot{\tilde{x}}_\Delta&=\eps^{-1}I\otimes(\bar{A}-\bar{B}F)\tilde{x}_\Delta+(I\otimes\bar{\Gamma})\tilde{x}_\Delta+\eps^{-1}(I\otimes(\bar{B}F))e_\Delta\label{wod1}\\
		\dot{e}_\Delta&= \eps^{-1}(I\otimes \bar{A}-\hat{L}\otimes I)e_\Delta+(I\otimes\bar{\Gamma})e_\Delta+\eps^{-1}I\otimes(\Delta\bar{\Delta}^{-1})\bar{e}_\Delta\label{wod3}\\
		\dot{\bar{e}}_\Delta&=(I\otimes \bar{A}_{\bar{e}})\bar{e}_\Delta\label{wod2}
	\end{align}
	
	We consider stability of \eqref{wod2} first. Consider the following Lyapunov function for \eqref{wod2}
	\[
	V_{\bar{e}}=\bar{e}_\Delta\T I\otimes \begin{pmatrix}\eps I&0\\0&P\end{pmatrix}\bar{e}_\Delta
	\]
	where $K_1 \geq \frac{1}{2}\alpha I$ and $P>0$ satisfies 
	\[
	P(\bar{A}_1+\bar{B}_1\Gamma_2-K_2\bar{C}_1)+(\bar{A}_1+\bar{B}_1\Gamma_2-K_2\bar{C}_1)\T P\leq -\beta I
	\]
	with $\alpha=1+2\|\bar{C}_1K_2\|+\|P\Psi\|^2+\|\bar{C}_1\|^2$ and $\beta=3$.
	
	Let $\bar{e}_\Delta=\begin{pmatrix}
		\bar{e}_{\Delta,1}\\\bar{e}_{\Delta,2}
	\end{pmatrix}$. Then, we have
	\begin{align*}
		\dot{V}_{\bar{e}}\leq&-\alpha \|\bar{e}_{\Delta,1}\|^2+2\eps \bar{e}_{\Delta,1}I\otimes (\bar{C}_1K_2)\bar{e}_{\Delta,1}+2\bar{e}_{\Delta,1}(I\otimes \bar{C}_1) \bar{e}_{\Delta,2}\\
		&\qquad-\beta \|\bar{e}_{\Delta,2}\|^2+2\eps \bar{e}_{\Delta,2}I\otimes (P\Psi) \bar{e}_{\Delta,1}\\
		\leq & - \|\bar{e}_{\Delta,1}\|^2- \|\bar{e}_{\Delta,2}\|^2
	\end{align*}
	for any $\eps\in (0,1]$.
	
	Thus, from $\dot{V}_{\bar{e}}<0$, we have $\bar{e}_\Delta$ is stable, and then we just need to consider the stability of 
	\begin{equation}\label{stwod1}
		\bar{A}-\bar{B}F +\eps \bar{\Gamma}
	\end{equation}
	and 
	\begin{equation}\label{stwod2}
		(I\otimes \bar{A}-\hat{L}\otimes I)+\eps (I\otimes\bar{\Gamma})
	\end{equation}
	for \eqref{wod1} and \eqref{wod3}, respectively. For \eqref{stwod2}, we just need to prove the stability of
	\begin{equation}
		I\otimes (\bar{A}+\bar{B}\Gamma)-\eps^{-1}\hat{L}\otimes I=I\otimes A_d-\eps^{-1}\hat{L}\otimes I
	\end{equation}
	based on the transformation between \eqref{systmat2} and \eqref{wod3}.
	According to the proof of \cite[Theorem 1]{liu-saberi-stoorvogel-donya-almost-automatica}, we have $I\otimes A_d-\hat{L}\otimes I$ is asymptotically stable because the eigenvalues of $A_d$ are all on the imaginary axis. Thus, it is obvious that $I\otimes A_d-\eps^{-1}\hat{L}\otimes I$ is also asymptotically stable for any $\eps\in (0,1]$.

	On the other hand, since $\bar{A}-\bar{B}F$ is Hurwitz stable, there exist a sufficiently small $\eps^*\in(0,1]$ satisfying $\eps<\eps^*$ such that \eqref{stwod1} is stable for $\eps<\eps^*$.
	Therefore, the system \eqref{wod1}--\eqref{wod3} is stable. That means MAS \eqref{comMAS} can achieve output synchronizations which implies that the original heterogeneous MAS achieves output synchronization.\\
	
	4)  \textit{$H_\infty$ almost synchronization in the presence of disturbance:}
	
	Finally, in this step, we prove $H_\infty$ almost disturbance rejection of the output $\bar{y}$ of the system \eqref{wd1}--\eqref{wd2} in presence of $\bar{\omega}$ which implies $H_\infty$ almost output synchronization. First, we consider $\bar{e}_\Delta$. For system \eqref{wd2}, we have
	\[
	\left\|\bar{e}_\Delta\T \left(\hat{L}\Pi\otimes \begin{pmatrix}
		\eps I&0\\0&P
	\end{pmatrix}\bar{\Delta}\right)\bar{E}\bar{\omega}\right\|\leq \eps \bar{\alpha} \|\bar{e}_\Delta\|\|\bar{\omega}\|
	\]
	where $\bar{\alpha}\geq \left\|\left(\hat{L}\Pi\otimes \begin{pmatrix}
		I&0\\-PK_2&P
	\end{pmatrix}\bar{\Delta}\right)\bar{E}\right\|$. 
	Then, from the result of $\dot{V}_{\bar{e}}$, we have
	\begin{align*}
		\dot{V}_{\bar{e}_{\Delta}}\leq &- \|\bar{e}_{\Delta}\|^2+2\eps \bar{\alpha} \|\bar{e}_\Delta\|\|\bar{\omega}\|\\
		\leq&-\frac{1}{2}\|\bar{e}_{\Delta}\|^2+2\eps^2 \bar{\alpha}^2\|\bar{\omega}\|^2
	\end{align*}

	Next consider a Lyapunov function 
	\[
	V_{e_{\Delta}}=\eps e_{\Delta}\T Q e_{\Delta}
	\]
	with 
	\[
	Q(I\otimes \bar{A}-\hat{L}\otimes I)+(I\otimes \bar{A}-\hat{L}\otimes I)\T Q<-4 I
	\]
	since $I\otimes \bar{A}-\hat{L}\otimes I$ is stable.
	
	We have
	\begin{align*}
		\dot{V}_{e_{\Delta}}\leq& -4\|e_{\Delta}\|^2+2\eps \|Q(I\otimes \bar{\Gamma})\|\|e_{\Delta}\|^2\\
		&\qquad+2 \|Q\Delta\bar{\Delta}^{-1}\| \|e_{\Delta}\|\|\bar{e}_{\Delta}\|+2\eps \bar{\beta} \|e_{\Delta}\|\|\bar{\omega}\| \\
		\leq&-\|e_{\Delta}\|^2+\|Q\Delta\bar{\Delta}^{-1}\|^2 \|\bar{e}_{\Delta}\|^2+\eps^2 \bar{\beta}^2 \|\bar{\omega}\|^2
	\end{align*}
	for any $\eps\in(0,\eps_1]$ with $\eps_1\in(0,\eps^*]$ satisfying
	\[
	2\eps_1 \|Q(I\otimes \bar{\Gamma})\|<1,
	\]
	and $\bar{\beta}\geq \|Q(\Pi\otimes \Delta)\bar{E}\|$. Here we have $\|Q\Delta\bar{\Delta}^{-1}\|$ is bounded as the function of $\eps$.
	
	Similarly, for \eqref{wd1}, we consider the following 
	Lyapunov function 
	\[
	V_{\tilde{x}_\Delta}=\eps \tilde{x}_\Delta\T (I\otimes Z) \tilde{x}_\Delta
	\]
	with 
	\[
	Z(\bar{A}-\bar{B}F )+(\bar{A}-\bar{B}F )\T Z<-4 I
	\]
	since $\bar{A}-\bar{B}F$ is Hurwitz stable.
	Then we have
	\begin{align*}
		\dot{V}_{\tilde{x}_\Delta}\leq& -4\|\tilde{x}_\Delta\|^2+2\eps \|Z \bar{\Gamma}\|\|\tilde{x}_\Delta\|^2\\
		&\qquad+2 \|Z\bar{B}F\| \|\tilde{x}_\Delta\|\|e_{\Delta}\|+2\eps \bar{\delta} \|\tilde{x}_\Delta\|\|\bar{\omega}\| \\
		\leq&-\|\tilde{x}_\Delta\|^2+\|Z\bar{B}F\|^2 \|e_{\Delta}\|^2+\eps^2 \bar{\delta} ^2 \|\bar{\omega}\|^2
	\end{align*}
	for any $\eps\in(0,\eps_2]$ with $\eps_2\in(0,\eps^*]$ satisfying
	\[
	2\eps_2 \|Z\bar{\Gamma}\|<1,
	\]
	and $\bar{\delta}\geq \|(\Pi\otimes (Z\Delta))\bar{E}\|$.
	
	Thus, let
	\[
	V=2\|Z\bar{B}F\|^2\|Q\Delta\bar{\Delta}^{-1}\|^2{V}_{\bar{e}_{\Delta}}+\|Z\bar{B}F\|^2V_{e_{\Delta}}+V_{\tilde{x}_\Delta}
	\]
	then we have
	\begin{align}
		\nonumber\dot{V}\leq& -\|Z\bar{B}F\|^2\|Q\Delta\bar{\Delta}^{-1}\|^2\|\bar{e}_{\Delta}\|^2+4\eps^2 \tilde{\alpha}^2\|\bar{\omega}\|^2\\
		\nonumber&-\|Z\bar{B}F\|^2\|e_{\Delta}\|^2+\|Z\bar{B}F\|^2\|Q\Delta\bar{\Delta}^{-1}\|^2 \|\bar{e}_{\Delta}\|^2+\eps^2 \tilde{\beta}^2 \|\bar{\omega}\|^2\\
		\nonumber&-\|\tilde{x}_\Delta\|^2+\|Z\bar{B}F\|^2 \|e_{\Delta}\|^2+\eps^2 \bar{\delta}^2 \|\bar{\omega}\|^2\\
		=&-\|\tilde{x}_\Delta\|^2+\eps^2 \gamma^2\|\bar{\omega}\|^2\label{bound}
	\end{align}
	for any $\eps\in(0,\bar\eps]$ with $\bar\eps=\min(\eps_1,\eps_2)$, suitable $\tilde{\alpha}$ and $\tilde{\beta}$, and $\gamma=\sqrt{4\tilde{\alpha}^2+\tilde{\beta}^2+\bar{\delta}^2}$.

	Since $\tilde{y}=(I\otimes \bar{C})\tilde{x}_{\Delta}$ and $\|\bar{C}\|=1$, we have $\|\tilde{y}\|\leq \|\tilde{x}_{\Delta}\|$. Then, we obtain
	\[
	\dot{V}+\|\tilde{y}\|^2-\eps^2 \gamma^2\|\bar{\omega}\|^2\leq 0
	\]
	Thus we have
	\[
	\|T_{\bar{\omega}\tilde{y}}\|_{H_\infty}\leq \eps \gamma
	\]
	i.e.
	\[
	\|T_{\bar{\omega}(y_i-y_j)}\|_{H_\infty}\leq \eps {\gamma}
	\]
	with $0<\eps <\bar\eps$ and $\gamma>0$. 
\end{proof}

\begin{remark}
	Following the instruction given in \cite[Appendix B]{yang-saberi-stoorvogel-grip-journal}, the compensator design includes three steps:
	\begin{enumerate}
		\item Squaring down precompensator;
		\item Rank-equalizing precompensator;
		\item Observer-based pre-feedback.
	\end{enumerate}
	where the design procedures for step 1 and 2 were developed in \cite{sannuti-saberi} and \cite{sks,sannuti-saberi-zhang-auto}, respectively. It is worth noting that the process of designing precompensators does not introduce zeros in the closed right half plane.
\end{remark}

\section{Scale-free $H_\infty$ almost regulated output synchronization}

According to Figure \ref{Architecture-reg}, the design for scale-free $H_{\infty}$ almost regulated output synchronization includes two major modules. First, we design precompensators for homogenizing agents, then we design collaborative protocols for the compensated agents to achieve $H_{\infty}$ almost regulated output synchronization. The design consists of the following steps.

\hspace{3mm} \textbf{Step I: remodeling the exosystem and choosing target model} 

Following the design procedure in \cite[Appendix C]{yang-saberi-stoorvogel-grip-journal} we remodel the exosystem \eqref{exo} to arrive at a suitable choice for the target model. There exists another exosystem given by
\begin{equation}\label{exo-2}
\begin{system*}{cl}
\dot{\check{x}}_r&=\check{A}_r\check{x}_r, \quad \check{x}_r(0)=\check{x}_{r0}\\
y_r&=\check{C}_r\check{x}_r,
\end{system*}
\end{equation}
such that for all $x_{r0} \in \mathbb{R}^r$, there exists $\check{x}_{r0}\in \mathbb{R}^{\check{r}}$ for which \eqref{exo-2} generate exactly the same output $y_r$ as the original exosystem \eqref{exo}. Furthermore, we can find a matrix $\check{B}_r$ such that the triple $(\check{C}_r,\check{A}_r,\check{B}_r)$ is invertible, of uniform rank $n_q$, and has no invariant zero, where $n_q$ is an integer greater than or equal to maximal order of infinite zeros of $(C_i,A_i,B_i), i\in \{1,...,N\}$ and all the observability indices of $(C_r, A_r)$. Note that the eigenvalues of $\check{A}_r$ consists of all eigenvalues of $A_r$ and additional zero eigenvalues as such eigenvalues of $\check{A}_r$ are in the closed left have plane. We choose our target model as $(\check{C}_r,\check{A}_r,\check{B}_r)$.

\begin{remark}\label{remark-ABC-exo}
Similar to Remark \ref{remark-ABC}, we assume the $(\check{C}_r,\check{A}_r,\check{B}_r)$ has the following form.
	\[
	\begin{system*}{cll}
	\check{A}_r&=\bar{A}+\bar{B}\Gamma, \bar{A}:=\begin{pmatrix}
	0&I_{p(n_q-1)}\\0&0	\end{pmatrix}, \check{B}_r=\bar{B}:=\begin{pmatrix}
	0\\I_p
	\end{pmatrix},\\
	\check{C}_r&=\bar{C}:=\begin{pmatrix}
	I_p&0
	\end{pmatrix}.
	\end{system*}
	\]
\end{remark}

\hspace{3mm} \textbf{Step II: designing precompensators}  

Following the instruction given in \cite[Appendix B]{yang-saberi-stoorvogel-grip-journal}, we use the following precompensator for each agent of MAS \eqref{homst-agent-model} to homogenize the agents to the target model chosen in the previous step.
\begin{tcolorbox}[colback=white]
	\begin{equation}\label{precom2}
		\begin{system*}{cl}
			\dot{p}_i&=G_ip_i+H_{1i}v_i+H_{2i}z_{i}, \\
			u_i&=Q_ip_i+R_{1i}v_i+R_{2i}z_{i},
		\end{system*}
	\end{equation}
	where $v_i$ is the input of the precompensator.
\end{tcolorbox}

We obtain the compensated agents by combining \eqref{homst-agent-model} and \eqref{precom2} as
\begin{equation}\label{comMAS2}
	\begin{system*}{cl}
		\dot{\bar{x}}_i&=\bar{A}\bar{x}_i+\bar{B}(v_i+\Gamma\bar{x}_i)+\bar{E}_i\bar{\omega}_i \\
		y_i&=\bar{C}\bar{x}_i.
	\end{system*}
\end{equation}

\hspace{3mm} \textbf{Step III: designing collaborative protocols for the compensated agents} Finally, the following linear dynamic protocol is designed for the compensated agents \eqref{comMAS2} as
\begin{tcolorbox}[breakable,colback=white]
	\begin{equation}\label{pro-lin-partial2}
		\begin{system}{cll}
			\dot{\hat{x}}_i&=& \bar{A}\hat{x}_i-\eps^{-n_q}\bar{B}F\Delta\hat{\zeta}_{i}+\bar{B}\Gamma\hat{x}_i\\
			&&\hspace{2cm}+\eps^{-1}K(\bar{\zeta}_{i}-C\hat{x}_i)+\iota_i \bar{B}v_i\\
			\dot{\chi}_i &=& \bar{A}\chi_i+\bar{B}v_i+\bar{B}\Gamma\chi_i+\eps^{-1}(\hat{x}_i-\hat{\zeta}_{i})-\eps^{-1}\iota_i\chi_i\\
			v_i &=& -\eps^{-n_q}F\Delta\chi_i,
		\end{system}
	\end{equation}
	where $\eps\in(0,\eps^*]$ is the tuning parameter. We define the scaling matrix $\Delta=\diag\{I_p,\eps I_p,\hdots, \eps^{n_q-1}I_p\}\in\mathbb{R}^{pn_q \times pn_q}$ and we choose matrix $F$ such that $\bar{A}-\bar{B}F$ is Hurwitz stable. Meanwhile, we partition $\bar{A}, \bar{B}, \Gamma$ as
	\[
	\bar{A}=\begin{pmatrix}
		0_{p\times p}&\bar{C}_1\\0_{p(n_q-1)\times p}&\bar{A}_1
	\end{pmatrix}, \bar{B}=\begin{pmatrix}
		0_{p\times p}\\\bar{B}_1
	\end{pmatrix}, \Gamma=\begin{pmatrix}\Gamma_1&\Gamma_2\end{pmatrix}
	\]
	with $\bar{C}_1=\begin{pmatrix}
		I_p&0&\cdots&0
	\end{pmatrix}$.
	We  choose $K=\begin{pmatrix}
		K_1\\K_2K_1
	\end{pmatrix}$
	with $K_2$ such that $\bar{A}_1+\bar{B}_1\Gamma_2-K_2\bar{C}_1$ is Hurwitz stable. And then let $K_1=K_1\T>\frac{\alpha}{2}I\in \R^{p\times p}$, where the value of $\alpha$ depends on $K_2$ and is given explicitly in the proof of Theorem \ref{sstem1}.
	
	In this protocol, each agent communicate $\chi_i$ with its neighbors, (i.e., $\xi_i=\chi_i$). Therefore, each agent has access to the localized information exchange appeared in \eqref{pro-lin-partial2} represented by \eqref{add_1},
	which can be simply obtained by $\xi_i=\chi_i$ in \eqref{eqa1}.
\end{tcolorbox}
We have the following theorem.

\begin{theorem}\label{aros1}
	Consider a MAS described by \eqref{homst-agent-model} and  \eqref{zetabar2}, and the exosystem \eqref{exo} satisfying Assumption \ref{ass2} and \ref{ass1}. Let $\mathbb{G}^N_\mathscr{C}$ be the set of network graphs as
	defined in Definition \ref{def_rootset}. Then, the scalable $H_\infty$-AROSWLIE problem as stated in Problem \ref{prob-reg} is solvable. More specifically,
	\begin{enumerate}
		\item  in the absence of the disturbance $\omega$ protocol \eqref{pro-lin-partial2} and \eqref{precom2} achieves scalable regulated output synchronization \eqref{synch_org}, for any graph $\mathcal{G}\in \mathbb{G}^N_\mathscr{C}$ with any number of agents $N$, and for any $\eps\in(0,\eps^*]$ with $\eps^*<1$ where the value of $\eps^*$ only depends on $\bar{A},\bar{B},\Gamma, F$,
		\item in the presence of the disturbance
		$\omega$, for any $\gamma>0$, the $H_\infty$ norm from $\omega$ to
		$y_i-y_r$ is less that $\gamma$ for all $i\in\{1,\hdots, N\}$ by choosing $\eps$ sufficiently small.
	\end{enumerate}
\end{theorem}

\begin{proof}
		Similar to the proof of Theorem \ref{sstem1}, we just replace the $\hat{L}$ with the expanded Laplacian matrix $\bar{L}$ for error system $x_i-x_r$. Then, we obtain the $H_\infty$ almost regulated output synchronization result with arbitrary $H_\infty$ norm from $\bar{\omega}$ to $y_i-y_r$.
\end{proof}

\begin{remark}
	It is worth noting that the formulations in Problem \ref{prob4} and \ref{prob-reg} do not specify how the matrices of the protocols should evolve with parameter $\eps$. However, our designs, as given in Protocol \eqref{pro-lin-partial} and \eqref{pro-lin-partial2}, provide an explicit solution from which the protocol matrices can be derived. The structure of the protocols are independent of the parameter $\eps$; thus, one may develop the structure at one stage and tune the parameter $\eps$ later so as to obtain the desired degree of accuracy. Due to continuity in $\eps$, tuning may be even carried out online. Hence, the method is a one-shot design and is not iterative.
\end{remark}

\section{Simulation Results}
In this section, we will illustrate the effectiveness of our protocols with numerical examples for $H_\infty$ almost output synchronization of heterogeneous MAS. We show that our one-shot-designed protocol \eqref{pro-lin-partial} is scale-free and works for any MAS with any communication graph and any number of agents. Consider the agents models \eqref{homst-agent-model} with
\begin{equation*}
A_i=\begin{pmatrix}
0&1&0&0\\0&0&1&0\\0&0&0&1\\0&0&0&0
\end{pmatrix}, B_i=\begin{pmatrix}
0&1\\0&0\\1&0\\0&1
\end{pmatrix}, C_i\T=\begin{pmatrix}
1\\0\\0\\0
\end{pmatrix},E_i=\begin{pmatrix}
1\\1\\0\\0
\end{pmatrix},\\
\end{equation*}
and $C^m_i=I$, for $i=1,6$, and
\begin{equation*}
A_i=\begin{pmatrix}
0&1&0\\0&0&1\\0&0&0
\end{pmatrix},B_i=\begin{pmatrix}
0\\0\\1
\end{pmatrix},C_i\T=\begin{pmatrix}
1\\0\\0
\end{pmatrix},E_i=\begin{pmatrix}
1\\1\\0
\end{pmatrix},
C^m_i=I,
\end{equation*}
for $i=2,7$, and
\begin{equation*}
A_i=\begin{pmatrix}
-1&0&0&-1&0\\0&0&1&1&0\\0&1&-1&1&0\\0&0&0&1&1\\-1&1&0&1&1
\end{pmatrix},B_i=\begin{pmatrix}
0&0\\0&0\\0&1\\0&0\\1&0
\end{pmatrix},C_i\T=\begin{pmatrix}
0\\0\\0\\1\\0
\end{pmatrix},
\end{equation*}
and $E_i=\begin{pmatrix}
1&1&0&0&0
\end{pmatrix}\T$, $C^m_i=I,$ for $i=3,4,8,9$. Finally
\begin{equation*}
A_i=\begin{pmatrix}
0&1&0\\0&0&1\\1&1&0
\end{pmatrix},B_i=\begin{pmatrix}
0\\0\\1
\end{pmatrix},C_i\T=\begin{pmatrix}
1\\0\\0
\end{pmatrix},E_i=\begin{pmatrix}
1\\0\\0
\end{pmatrix},
C^m_i=I,
\end{equation*}
for $i=5,10$. The agents are subject to disturbances $\omega_i=sin(t)$, for $i=2,5,7,10$, $\omega_i=sin(2t)$ for $i=3,4,8,9$, and $\|\omega_i\|\leq5$ for $i=1,6$ which is a random number normally distributed.

To show the scalability of our protocols, we consider two heterogeneous MAS with different number of agents and different communication topologies.

\textbf{Case $1$:} Consider a MAS with $4$ agents with agent models $(C_i, A_i, B_i)$ for $i \in\{ 1,\hdots,4\}$, and directed communication topology with adjacency matrix $\mathcal{A}_1$ where $a_{14}=a_{21}=a_{31}=a_{42}=1$.

\textbf{Case $2$:} Next, consider a MAS with $10$ agents with agent models $(C_i, A_i, B_i)$ for $i \in \{1,\hdots,10\}$ and directed communication topology with adjacency matrix $\mathcal{A}_2$, where $a_{21}=a_{5,10}=a_{32}=a_{43}=a_{54}=a_{65}=a_{76}=a_{87}=a_{98}=a_{10,9}=a_{15}=1$.

Note that for both cases $p=1$ and $\bar{n}_d=3$, which is the degree of infinite zeros of $(C_2,A_2,B_2)$. It is obtained that $n_q=3$. We choose the target model as
\[
A_d=\begin{pmatrix}
0&1&0\\0&0&1\\0&-1&0
\end{pmatrix}, B_d=\begin{pmatrix}0\\0\\1\end{pmatrix}, C_d=\begin{pmatrix}1&0&0\end{pmatrix}.
\] 

By designing precompensators as stated in \emph{step II}, we obtain the compensated agents \eqref{comMAS} with $\Gamma=\begin{pmatrix}
0&-1&0
\end{pmatrix}$. We choose $F=\begin{pmatrix}
30&31&10
\end{pmatrix}$ and $K\T=\begin{pmatrix}
1&1&1
\end{pmatrix}$.

The simulation results presented in Figure \ref{N4} and \ref{N10} show that the protocol design is independent of the communication graph and is scale-free so that we can achieve $H_{\infty}$ almost output synchronization via our one-shot-designed protocol, for any graph with any number of agents. The simulation results also show that by decreasing the value of $\eps$, $H_\infty$ almost output synchronization is achieved with higher degree of accuracy.

\begin{figure}[t]
	\includegraphics[width=9cm, height=5cm]{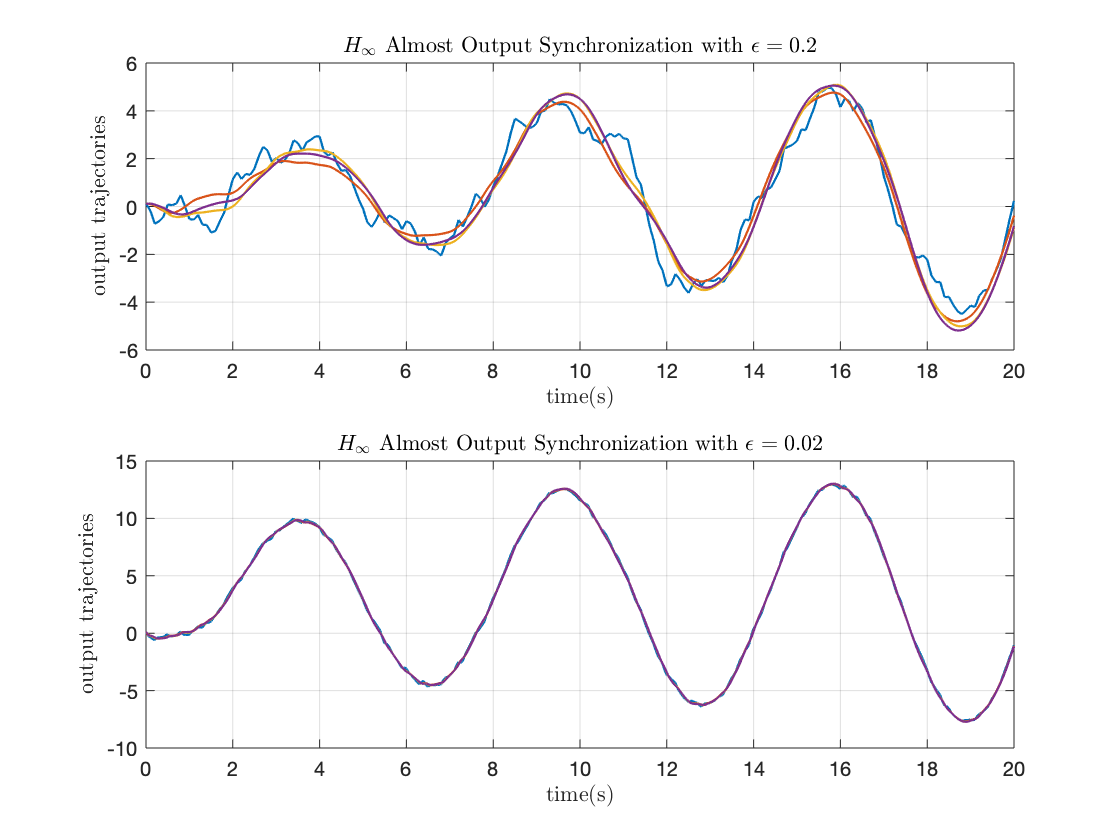}
	\centering
			\vspace{-5mm}
	\caption{$H_{\infty}$ Almost output synchronization for MAS of case $1$ with $N=4$}\label{N4}
	\vspace{-2mm}
\end{figure} 
\begin{figure}[t!]
	\includegraphics[width=9cm, height=5cm]{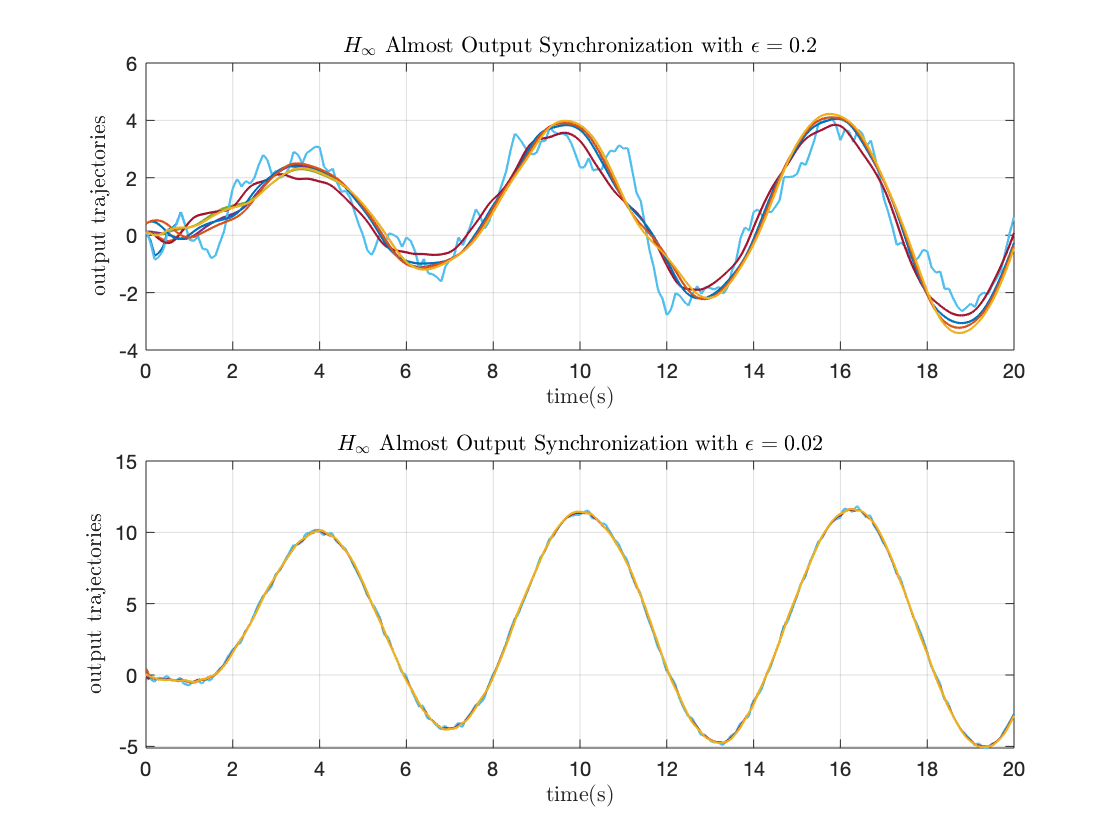}
	\centering
			\vspace{-5mm}
	\caption{$H_{\infty}$ Almost output synchronization for MAS of case $2$ with $N=10$}\label{N10}
		\vspace{-2mm}
\end{figure} 

\bibliographystyle{plain}
\bibliography{referenc}
\end{document}